\documentclass[aps,prl,reprint,superscriptaddress,longbibliography]{revtex4-2}
\usepackage{CJK}
\usepackage[T1]{fontenc}
\usepackage{microtype}
\usepackage{amsmath,amssymb,bm,graphicx,tikz}
\usepackage{amsthm}
\theoremstyle{plain}
\newtheorem{theorem}{Theorem}
\newtheorem{corollary}{Corollary}
\theoremstyle{remark}
\newtheorem{note}{Note}
\usepackage[colorlinks=true,linkcolor=blue!50!black,citecolor=blue!50!black,urlcolor=blue!50!black]{hyperref}
\hypersetup{pdftitle={Exact Strong Symmetry from Two-Point Correlations},pdfsubject={Version 12: exact strong symmetry and limits on spontaneous order}}
\newcommand{\id}{\mathbb I}
\newcommand{\Tr}{\operatorname{Tr}}
\newcommand{\Var}{\operatorname{Var}}

\newcommand{\ket}[1]{\lvert #1\rangle}
\newcommand{\bra}[1]{\langle #1\rvert}
\newcommand{\norm}[1]{\lVert #1\rVert}
\newcommand{\SU}{\mathrm{SU}}
\newcommand{\Uone}{\mathrm U(1)}

\begin{document}
\begin{CJK*}{UTF8}{gbsn} 
\title{Strong Symmetry from Two-Point Correlations}
\date{September 19, 2026}
\author{Han Yan (闫寒)}
\email{hanyan@issp.u-tokyo.ac.jp}
\affiliation{Institute for Solid State Physics, The University of Tokyo, Kashiwa, Chiba 277-8581, Japan}

\begin{abstract}
Strong and weak symmetries and their spontaneous breaking distinguish different forms of symmetry and order in mixed states. We show that complete one- and two-point correlation functions, equivalently all two-site reduced density matrices, determine exact strong symmetry for continuous onsite unitary actions of connected groups: states with identical correlations have the same strong symmetry and character. 
We also present an algorithm to determine the maximal connected strong-symmetry sub-Lie group and its character from two-point correlations. 
Because such correlations are more accessible than full-state tomography, the algorithm is readily applicable to experiments ranging from condensed-matter systems to quantum simulators and circuits. We also identify group and representation-specific cases in which fewer correlations suffice.
\end{abstract}
\maketitle
\end{CJK*}

\section{Introduction}
\label{sec:intro}
Symmetry constrains the density matrix of a mixed state in two distinct ways. Weak symmetry leaves the ensemble invariant under conjugation, whereas strong symmetry requires every state in its support to transform with the same symmetry phase~\cite{LessaAnomaly2025}. Strong and weak symmetries also constrain Lindblad dynamics~\cite{Buca2012,Albert2014} and mixed-state symmetry-protected topological order~\cite{deGroot2022,MaTurzillo2025}. The distinction for states underlies symmetry-enforced entanglement in mixed singlet ensembles~\cite{Moharramipour2024,Li2025} and strong-to-weak spontaneous symmetry breaking, whose order is characterized by nonlinear correlation functions~\cite{Lessa2025,Sala2024}. Such ordering transitions have been studied in dissipative models~\cite{Gu2025} and decohered spin systems~\cite{Orito2025}.

In many-body experiments, full-state tomography rapidly becomes impractical, while overlapping tomography, classical shadows, and collective measurements can access large families of pair RDMs or two-point correlations~\cite{Cotler2020,Huang2020,Yang2023,Toth2007,Urizar2013,Vitagliano2025}. This restriction raises a theoretical question with practical relevance: assuming access to all one- and two-point correlations, or equivalently all two-site reduced density matrices (RDMs), can such low-order data determine whether the exact strong-symmetry constraint holds?

Two-site RDMs generally leave the global density matrix undetermined~\cite{Ocko2011,Walck2008}. For continuous onsite symmetries, however, total charges are additive, and their fluctuations involve only one- and two-point functions. We use this structure to show that identical pair RDMs imply the same exact strong symmetry and symmetry charge for any continuous onsite unitary action of a connected group. Thus low-order correlations fix a global property of the state without reconstructing it. For Lie groups, we further construct the connected strong-symmetry subgroups from the charge covariance matrix and discuss cases in which fewer correlations suffice. This identifies the exact strong-symmetry information contained in experimentally accessible low-order data.

\section{Theorems on strong symmetry from two-point correlations}
\label{sec:new_statements}

We consider a general quantum many-body setup: a finite system of $N\ge2$ sites with Hilbert space $\mathcal H=\bigotimes_{i=1}^N\mathcal H_i$, where each $\mathcal H_i$ is finite dimensional. Let $G$ be a connected topological group with an ordinary continuous onsite unitary representation
\begin{equation}
 U(g)=\bigotimes_{i=1}^N u_i(g),\qquad g\in G.
 \label{eq:new_onsite}
\end{equation}
Each $u_i$ is a continuous unitary representation on $\mathcal H_i$. A density matrix $\rho$ has strong $G$ symmetry if $U(g)\rho=\chi(g)\rho$ for every $g\in G$, where $\chi:G\to\mathrm U(1)$ is a continuous one-dimensional representation, or character.

States with identical two-point correlations must have the same exact continuous strong symmetry. The proofs of the theorem and corollary follow in the next section.

\begin{theorem}[Strong symmetries from two-point correlations]
\label{thm:connected}
Let $\rho$ and $\sigma$ be density matrices on the same Hilbert space $\mathcal H$, with the same onsite action $U(g)$. Suppose that their two-point correlation functions are identical:
\begin{equation}
 \Tr(\rho A_iB_j)=\Tr(\sigma A_iB_j)
 \quad\text{for every }i\ne j,
 \label{eq:new_correlations}
\end{equation}
for all single-site observables $A_i$ and $B_j$, including the identity. Then, for every continuous character $\chi$ of $G$,
\begin{equation}
 \begin{gathered}
 U(g)\rho=\chi(g)\rho\\
 \text{for all }g\in G
 \end{gathered}
 \quad\Longleftrightarrow\quad
 \begin{gathered}
 U(g)\sigma=\chi(g)\sigma\\
 \text{for all }g\in G.
 \end{gathered}
 \label{eq:new_strong_equivalence}
\end{equation}
Thus strong symmetry under this group action is either present in both states, with the same character, or absent in both. In particular, if no character satisfies the strong-symmetry condition for $\rho$, none does for $\sigma$, and conversely.

\end{theorem}

\begin{corollary}[Two-site reduced density matrices]
\label{cor:pair_rdms}
For each pair of sites $i,j$, define the two-site reduced density matrices (RDMs)
\begin{equation}
 \rho_{ij}=\Tr_{\overline{ij}}\rho,\qquad
 \sigma_{ij}=\Tr_{\overline{ij}}\sigma,
 \label{eq:new_rdms}
\end{equation}
where $\overline{ij}$ denotes all sites other than $i$ and $j$. These operators describe the states restricted to $\mathcal H_i\otimes\mathcal H_j$. Equality of the complete two-point correlation functions, with the identity factors included as above, is equivalent to equality of all pair RDMs:
\begin{equation}
 Equation~\eqref{eq:new_correlations}
 \quad\Longleftrightarrow\quad
 \rho_{ij}=\sigma_{ij}\quad\text{for all }i<j.
 \label{eq:new_rdm_equivalence}
\end{equation}
Consequently, under the assumptions of Theorem~\ref{thm:connected}, identical pair RDMs imply the same strong symmetry and the same character.

\end{corollary}

\begin{note}[Strong discrete symmetries cannot be determined.]
\label{note:discrete}
Theorem~\ref{thm:connected} does not extend to discrete symmetry groups in general: identical pair RDMs need not preserve either strong symmetry or its character. The following counterexample illustrates this limitation. On $N\ge3$ qubits, let $Z_i=\ket0\bra0_i-\ket1\bra1_i$ and let the nontrivial element of $\mathbb Z_2$ act as $P=\bigotimes_iZ_i$. Consider
\begin{equation}
 \rho_+=\frac{\id+P}{2^N},\qquad
 \rho_-=\frac{\id-P}{2^N},\qquad
 \tau=\frac{\id}{2^N}.
 \label{eq:new_parity_states}
\end{equation}
These are density matrices since $P^2=\id$ and $\Tr P=0$. Tracing out any site removes the parity term, since $\Tr Z_i=0$. Every proper subset $R$ therefore has
\begin{equation}
 (\rho_+)_R=(\rho_-)_R=\tau_R=\frac{\id_R}{2^{|R|}}.
 \label{eq:new_parity_rdms}
\end{equation}
All pair RDMs therefore coincide, yet
\begin{equation}
 P\rho_+=\rho_+,\qquad P\rho_-=-\rho_-,\qquad
 P\tau\ne\pm\tau.
 \label{eq:new_parity_symmetry}
\end{equation}
The states $\rho_\pm$ carry opposite strong characters, whereas $\tau$ has only weak $\mathbb Z_2$ symmetry, $P\tau P^\dagger=\tau$.
\end{note}

\begin{note}[Weak symmetry]
\label{note:weak}
Identical pair RDMs need not preserve weak symmetry, even for a connected group. On $N\ge3$ qubits, take $U(\theta)=e^{i\theta Q}$ with $Q=\sum_i\ket1\bra1_i$. Writing $\ket{0^N}=\ket0^{\otimes N}$ and $\ket{1^N}=\ket1^{\otimes N}$, consider
\begin{equation}
 \begin{aligned}
 \ket{\psi}&=\frac{\ket{0^N}+\ket{1^N}}{\sqrt2},\qquad
 \rho=\ket\psi\bra\psi,\\
 \sigma&=\tfrac12\bigl(\ket{0^N}\bra{0^N}+\ket{1^N}\bra{1^N}\bigr).
 \end{aligned}
 \label{eq:weak_counterexample}
\end{equation}
Tracing out any site removes the coherence between $\ket{0^N}$ and $\ket{1^N}$, so all pair RDMs coincide. However, $[Q,\sigma]=0$ while $[Q,\rho]\ne0$: only $\sigma$ has weak $\Uone$ symmetry. Both states lack strong $\Uone$ symmetry because their charge fluctuates. Thus Theorem~\ref{thm:connected} fixes the absence of strong symmetry while leaving weak symmetry undetermined.
\end{note}

\section{Proofs of Theorem~\ref{thm:connected} and Corollary~\ref{cor:pair_rdms}}
\label{sec:new_proofs}

\begin{proof}[Proof of Theorem~\ref{thm:connected}]
The essential fact is that a definite total charge has no fluctuations, while the squared fluctuation of an additive charge contains only single-site and two-site terms. We first prove the result for a connected Lie group and then extend it to the general connected topological group in the statement.

Choose a basis $X_a$ of the Lie algebra. The onsite representation gives additive Hermitian generators $Q_a=\sum_iq_i^a$, where $q_i^a$ acts only on site $i$. For a fixed continuous character $\chi$, write
\begin{equation}
 U(e^{tX_a})=e^{itQ_a},\qquad
 \chi(e^{tX_a})=e^{itc_a},
 \label{eq:new_generators}
\end{equation}
with real $c_a$. If $\rho$ has strong symmetry with character $\chi$, differentiating at $t=0$ gives $Q_a\rho=c_a\rho$ for every $a$. Hence the positive operator
\begin{equation}
 W_\chi=\sum_a(Q_a-c_a\id)^2
 \label{eq:new_fluctuation}
\end{equation}
has zero expectation in $\rho$. Its expectation is determined by the assumed correlations, because
\begin{equation}
 Q_a^2=\sum_i(q_i^a)^2+2\sum_{i<j}q_i^aq_j^a.
 \label{eq:new_charge_square}
\end{equation}
The first sum is fixed by one-point functions and the second by two-point functions. The linear terms in $W_\chi$ are also fixed by one-point functions. It follows that $\Tr(\sigma W_\chi)=\Tr(\rho W_\chi)=0$.

Positivity turns this vanishing expectation into an exact charge constraint. Indeed,
\begin{equation}
 0=\Tr(\sigma W_\chi)
   =\sum_a\norm{(Q_a-c_a\id)\sqrt\sigma}_2^2,
 \label{eq:new_positive_sum}
\end{equation}
where $\norm A_2^2=\Tr(A^\dagger A)$. Every term must vanish, giving $Q_a\sigma=c_a\sigma$. Exponentiation yields $U(e^{tX_a})\sigma=\chi(e^{tX_a})\sigma$. The one-parameter subgroups associated with a Lie-algebra basis generate the connected group, so $U(g)\sigma=\chi(g)\sigma$ for every $g\in G$. Exchanging $\rho$ and $\sigma$ proves equivalence for each character and shows that either both states admit a strong character or neither does.

The conclusion also holds for connected $G$ not assumed to be a Lie group; only its finite-dimensional local action matters. With $d_i=\dim\mathcal H_i$, define
\begin{equation}
 K=\overline{\{(u_1(g),\ldots,u_N(g)):g\in G\}}
 \subseteq\prod_i\mathrm U(d_i).
 \label{eq:new_image_closure}
\end{equation}
The continuous image of a connected group and its closure are connected. As a closed subgroup of a finite product of unitary groups, $K$ is therefore a connected compact Lie group. It acts onsite through $V(k)=\bigotimes_i k_i$. If $U(g)\rho=\chi(g)\rho$, the function $\widetilde\chi(k)=\Tr[V(k)\rho]$ extends $\chi$ from the image of $G$ to $K$. By continuity, $V(k)\rho=\widetilde\chi(k)\rho$, and $\widetilde\chi$ is a multiplicative unit-modulus function, hence a continuous character. Applying the Lie-group result to $K$ gives $V(k)\sigma=\widetilde\chi(k)\sigma$. Restricting back to the image of $G$ gives the required equation with the original character $\chi$. The reverse direction follows by exchanging the two states.

\end{proof}

\begin{proof}[Proof of Corollary~\ref{cor:pair_rdms}]
A pair RDM is fully specified by the expectation values of products of single-site observables. To make this explicit, choose a Hermitian operator basis $\{T_i^\alpha\}_{\alpha=0}^{d_i^2-1}$ on each site, normalized by $T_i^0=\id_i$ and $\Tr(T_i^\alpha T_i^\beta)=d_i\delta_{\alpha\beta}$. The product operators form an orthogonal basis on the pair, so
\begin{equation}
 \rho_{ij}=\frac1{d_id_j}\sum_{\alpha,\beta}
 \Tr(\rho T_i^\alpha T_j^\beta)\,
 T_i^\alpha\otimes T_j^\beta,
 \label{eq:new_pair_expansion}
\end{equation}
and likewise for $\sigma_{ij}$. Equal two-point functions give equal coefficients in this expansion and therefore equal pair RDMs. Conversely, the definition of the partial trace gives
\begin{equation}
 \Tr(\rho A_iB_j)
 =\Tr[\rho_{ij}(A_i\otimes B_j)].
 \label{eq:new_partial_trace}
\end{equation}
Thus equal pair RDMs give equal two-point functions, including the one-point functions obtained with an identity factor. This proves Eq.~\eqref{eq:new_rdm_equivalence}; the strong-symmetry conclusion then follows from Theorem~\ref{thm:connected}.

\end{proof}

\section{Constructing strong-symmetry generators from two-point correlations}
\label{sec:constructive_generators}

For connected compact Lie groups, vanishing of the entire charge covariance matrix is known to be equivalent to strong symmetry~\cite{Kusuki2026}. We use its kernel to determine the maximal connected subgroup under which a state has strong symmetry. Related covariance-null-space methods infer parent Hamiltonians from correlation data~\cite{Qi2019,Chertkov2018,Petrovich2024} and local conservation laws from quantum dynamics~\cite{Zhan2024}; here the object determined is a symmetry of the state. The physical basis is that a definite total charge acts with one phase throughout the state's support.

Take a connected Lie group $G$ acting as in Eq.~\eqref{eq:new_onsite}, with Lie-algebra basis $X_a$ and additive Hermitian generators $Q_a=\sum_iq_i^a$. The state need not have strong $G$ symmetry; we seek all combinations of generators with definite charge.

The required charge fluctuations are obtained directly from the correlations. Write $\mu_a=\Tr(\rho Q_a)$ and define the real symmetric covariance matrix
\begin{equation}
\Gamma^{(\rho)}_{ab}
 =\tfrac12\Tr[\rho\{Q_a,Q_b\}]-\mu_a\mu_b.
 \label{eq:new_subgroup_covariance}
\end{equation}
This is the symmetrized generator covariance matrix used to characterize strong-symmetry breaking in Ref.~\cite{Kusuki2026}. Additivity expands $\{Q_a,Q_b\}$ into single-site anticommutators and products on distinct sites. Thus complete one- and two-point functions determine both $\Gamma^{(\rho)}$ and the means $\mu_a$.

To identify the subgroup determined by these fluctuations, collect all strong-symmetry operations into
\begin{equation}
 S_\rho=\{g\in G:U(g)\rho=z\rho
              \text{ for some }z\in\mathrm U(1)\}.
 \label{eq:new_strong_subgroup}
\end{equation}
Its identity component $H_\rho=S_\rho^0$ is the maximal connected strong-symmetry subgroup. We show that its generators are precisely the zero directions of $\Gamma^{(\rho)}$.

\begin{theorem}[Constructive characterization]
\label{thm:constructive_generators}
With the definitions above, write $X(v)=\sum_av_aX_a$ for real $v$. The Lie algebra of $H_\rho$ and the subgroup itself are
\begin{align}
 \mathfrak h_\rho&=\{X(v):v\in\ker\Gamma^{(\rho)}\},
 \label{eq:new_kernel_algebra}\\
 H_\rho&=\bigl\langle\exp_G X(v):v\in\ker\Gamma^{(\rho)}\bigr\rangle,
 \label{eq:new_kernel_group}
\end{align}
where angle brackets denote the subgroup generated by the indicated elements. The strong character on $H_\rho$ is fixed by
\begin{equation}
 \chi_\rho\!\left(\exp_G[tX(v)]\right)
 =e^{it\sum_av_a\mu_a},\qquad v\in\ker\Gamma^{(\rho)}.
 \label{eq:new_kernel_character}
\end{equation}
Every connected subgroup of $G$ under which $\rho$ has strong symmetry is contained in $H_\rho$. Consequently, complete one- and two-point correlations determine this maximal connected subgroup and its character.
\end{theorem}

\begin{proof}
First, the covariance kernel identifies exactly the charges that have a definite value throughout the state. For $Q(v)=\sum_av_aQ_a$ and $q(v)=\sum_av_a\mu_a$,
\begin{equation}
 v^{\mathsf T}\Gamma^{(\rho)}v
 =\Var_\rho[Q(v)]
 =\norm{[Q(v)-q(v)\id]\sqrt\rho}_2^2.
 \label{eq:new_subgroup_variance}
\end{equation}
The matrix is therefore positive semidefinite. Its quadratic form vanishes exactly on its kernel, while the squared norm vanishes exactly when $Q(v)$ acts as $q(v)$ on the support of $\rho$. Hence
\begin{equation}
 v\in\ker\Gamma^{(\rho)}
 \quad\Longleftrightarrow\quad
 Q(v)\rho=q(v)\rho.
 \label{eq:new_kernel_charge}
\end{equation}
Physically, every measurement of this total charge gives $q(v)$, although local charges may fluctuate. Equation~\eqref{eq:new_kernel_charge} is the direction-resolved form of the faithful zero-covariance criterion~\cite{Kusuki2026}.

A definite charge generates a continuous strong symmetry. Exponentiating Eq.~\eqref{eq:new_kernel_charge} gives
\begin{equation}
 U\!\left(\exp_G[tX(v)]\right)\rho=e^{itq(v)}\rho
 \quad\text{for all }t\in\mathbb R.
 \label{eq:new_kernel_flow}
\end{equation}
Conversely, differentiating the strong-symmetry condition along any one-parameter subgroup at $t=0$ gives the same definite-charge condition, with its charge fixed by taking the trace. Thus the kernel identifies every continuous strong-symmetry generator within the specified action. These generators close under commutators: if $A\rho=a\rho$ and $B\rho=b\rho$, then $i[A,B]\rho=0$.

These infinitesimal generators determine the entire connected subgroup. The set $S_\rho$ is a subgroup because the strong-symmetry condition is preserved under products and inverses. Its phase is uniquely given by $z(g)=\Tr[U(g)\rho]$; continuity then shows that the condition is also preserved under limits. Hence $S_\rho$ is closed and is a Lie subgroup of $G$. The equivalence just established identifies its Lie algebra as Eq.~\eqref{eq:new_kernel_algebra}. Its identity component is generated by its one-parameter subgroups, proving Eq.~\eqref{eq:new_kernel_group}. Any connected subgroup under which $\rho$ has strong symmetry lies in $S_\rho$ and contains the identity, so it lies in $H_\rho$. This proves maximality.

Finally, once the subgroup is known, the charge values fix its character. The phases multiply under group multiplication, so $z(g)$ restricts to a continuous character $\chi_\rho$ on $H_\rho$. Equation~\eqref{eq:new_kernel_flow} gives Eq.~\eqref{eq:new_kernel_character}. Since the one-parameter subgroups generate $H_\rho$, these phases uniquely determine the character on the whole connected subgroup. Both the covariance matrix and the charge means are fixed by the complete one- and two-point correlations, completing the construction.
\end{proof}

Identical correlations fix $\Gamma$ and $\mu$ and hence imply $H_\rho=H_\sigma$ with the same character, even when strong symmetry under the full group is absent. The kernel fixes the subgroup; the means fix its character. Disconnected components remain undetermined. For general connected topological groups, restricting Theorem~\ref{thm:connected} to each connected subgroup still gives equivalence; the explicit covariance construction requires a Lie algebra.

\section{Cases in which fewer two-point correlations are needed}
\label{sec:new_examples}

The converse of Theorem~\ref{thm:connected} is false: states with the same strong symmetry need not have the same complete two-point correlations. Moreover, selected correlations may suffice for a specified symmetry and representation. The $\Uone$ and $\SU(d)$ examples illustrate both points and show that the required correlation data can sometimes be reduced.

\textit{$\Uone$: same symmetry, different correlations.---}
For $N$ qubits, let $n_i=\ket1\bra1_i$, $Q=\sum_i n_i$, and $U(\theta)=e^{i\theta Q}$. At fixed $0<q<N$, compare the infinite-temperature state within the $q$-particle sector with the pure Dicke state:
\begin{equation}
 \begin{aligned}
 \rho_{\mathrm{mix}}&=\binom Nq^{-1}
   \sum_{|\boldsymbol n|=q}\ket{\boldsymbol n}\bra{\boldsymbol n},\\
 \rho_D&=\ket{D_q}\bra{D_q},\qquad
 \ket{D_q}=\binom Nq^{-1/2}
   \sum_{|\boldsymbol n|=q}\ket{\boldsymbol n}.
 \end{aligned}
 \label{eq:new_u1_states}
\end{equation}
Both have
\begin{equation}
 \langle n_i\rangle=\frac qN,\qquad
 \langle n_in_j\rangle=\frac{q(q-1)}{N(N-1)}\quad(i\ne j).
 \label{eq:new_u1_density_correlations}
\end{equation}
These relations give $\langle Q\rangle=q$ and $\langle Q^2\rangle=q^2$. Hence $\Var(Q)=0$, so positivity implies $Q\rho=q\rho$ and strong $\Uone$ symmetry with character $e^{iq\theta}$. Thus the occupation correlations in Eq.~\eqref{eq:new_u1_density_correlations}, rather than all two-point correlations, already establish the symmetry.

Strong symmetry does not fix phase coherence. With $S_i^+=\ket1\bra0_i$ and $i\ne j$,
\begin{equation}
 \langle S_i^+S_j^-\rangle_{\rho_{\mathrm{mix}}}=0,\qquad
 \langle S_i^+S_j^-\rangle_{\rho_D}
 =\frac{q(N-q)}{N(N-1)}.
 \label{eq:new_u1_coherence}
\end{equation}
The Dicke value remains finite at fixed filling, whereas the mixture has no phase coherence. The two states therefore share the same strong symmetry, character, and sufficient occupation correlations, but have different complete two-point correlations and pair RDMs. This is the converse failure of Theorem~\ref{thm:connected} and Corollary~\ref{cor:pair_rdms}.

\textit{Fundamental $\SU(d)$: one unmeasured pair allowed.---}
For $\mathcal H_i=\mathbb C^d$ and $u_i(g)=g$, strong $\SU(d)$ symmetry means that the state's support lies in the singlet subspace, because the group has no nontrivial continuous character. We assume $N=dn$, so that singlets exist, and denote their projector by $P_0$. Exchange correlations provide a particularly economical test of this support.

\begin{figure}[t]
\centering
\begin{tikzpicture}[x=1cm,y=1cm,
site/.style={circle,fill=white,draw=black,line width=.5pt,minimum size=4.5mm,inner sep=0pt,font=\footnotesize},
measured/.style={draw=blue!65!black,line width=.8pt},
missing/.style={draw=red!65!black,dash pattern=on 3pt off 1.5pt,line width=1.1pt}]
\begin{scope}[shift={(0,0)}]
\node[font=\small] at (0,1.5) {(a) $\SU(2)$, $N=4$};
\foreach \i/\angle in {1/135,2/45,3/-45,4/-135} {\coordinate (a\i) at (\angle:1.05);}
\foreach \i/\j in {1/3,1/4,2/3,2/4,3/4} {\draw[measured] (a\i)--(a\j);}
\draw[missing] (a1)--(a2);
\foreach \i in {1,...,4} {\node[site] at (a\i) {\i};}
\node[font=\footnotesize,align=center] at (0,-1.5) {Singlet or nonsinglet\\Same five pair RDMs};
\end{scope}
\begin{scope}[shift={(4,0)}]
\node[font=\small] at (0,1.5) {(b) $\SU(3)$, $N=6$};
\foreach \i/\angle in {1/120,2/60,3/0,4/-60,5/-120,6/180} {\coordinate (b\i) at (\angle:1.05);}
\foreach \i/\j in {1/3,1/4,1/5,1/6,2/3,2/4,2/5,2/6,3/4,3/5,3/6,4/5,4/6,5/6} {\draw[measured] (b\i)--(b\j);}
\draw[missing] (b1)--(b2);
\foreach \i in {1,...,6} {\node[site] at (b\i) {\i};}
\node[font=\footnotesize,align=center] at (0,-1.5) {Singlet support enforced\\State need not be unique};
\end{scope}
\draw[measured] (-1.15,-2.15)--(-.45,-2.15);
\node[anchor=west,font=\footnotesize] at (-.25,-2.15) {Specified singlet exchange correlations};
\draw[missing] (-1.15,-2.55)--(-.45,-2.55);
\node[anchor=west,font=\footnotesize] at (-.25,-2.55) {One unmeasured pair in each panel};
\end{tikzpicture}
\caption{Selected two-site data can suffice to establish strong symmetry. Blue edges denote specified data; the red dashed edge marks the unmeasured pair. Edges are correlation constraints, not interactions. (a) For four fundamental $\SU(2)$ spins, the five specified pair RDMs allow either singlet or nonsinglet support. (b) For six fundamental $\SU(3)$ spins, the fourteen specified exchange expectations enforce singlet support without fixing the state.}
\label{fig:new_pairs}
\end{figure}

Collective squared charges also underlie singlet and spin-squeezing criteria~\cite{Urizar2013,Toth2007,Vitagliano2025}. Normalize the generators by $\Tr(t^at^b)=\delta_{ab}/2$ and write $Q^a=\sum_i t_i^a$. If $F_{ij}$ exchanges sites $i$ and $j$, completeness gives
\begin{equation}
 C=\sum_a(Q^a)^2
 =NC_f\id+\sum_{i<j}\left(F_{ij}-\frac{\id}{d}\right),
 \label{eq:new_sud_Casimir}
\end{equation}
where $C_f=(d^2-1)/(2d)$. Its kernel is the singlet space, so all exchanges matching a singlet reference force $\langle C\rangle=0$ and strong symmetry~\cite{Moharramipour2024}.

For $d\ge3$, any one pair $A$ can be left unmeasured. Let $B$ be its complement, $\bm Q_{A,B}=\sum_{i\in A,B}\bm t_i$, and $C_{A,B}=\bm Q_{A,B}^2$. The operator
\begin{equation}
 W_B=\bm Q_B\cdot(\bm Q_A+\bm Q_B)
     =\frac{C+C_B-C_A}{2}
 \label{eq:new_sud_WB}
\end{equation}
contains no exchange within $A$ and is therefore fixed by the remaining exchange correlations. The $\SU(d)$ Casimir spectrum gives a nonsinglet gap $C\ge d$, together with $C_A^{\max}=(d-1)(d+2)/d$ and $C_B^{\min}=(d-2)(d+1)/d$~\cite{Fulton1991}. Since these operators commute,
\begin{equation}
 W_BP_0=0,\qquad
 W_B\big|_{P_0^\perp}\ge
 \frac{d+C_B^{\min}-C_A^{\max}}2=\frac{d-2}{2}.
 \label{eq:new_sud_omitbound}
\end{equation}
Agreement with the measured exchanges of a singlet reference gives $\langle W_B\rangle=0$ and, for $d\ge3$, forces singlet support. Thus the other exchanges establish strong symmetry without the pair $A$, as illustrated for $\SU(3)$ in Fig.~\ref{fig:new_pairs}(b). In a singlet, $\sum_{j\ne i}\langle F_{ij}\rangle=N/d-d$, so the omitted exchange is also recovered; an invariant fundamental pair state is fixed by this value~\cite{Werner1989}.

The restriction $d\ge3$ is essential. For fundamental $\SU(2)$ on even $N\ge4$, let the omitted pair be $A=(i,j)$, choose a singlet state $\omega_B$ on the other sites, and define $\ket s=(\ket{01}-\ket{10})/\sqrt2$. Then
\begin{equation}
 \sigma=(\ket{s}\bra{s})_A\otimes\omega_B,\qquad
 \rho=\frac{\id_A-(\ket{s}\bra{s})_A}{3}\otimes\omega_B,
 \label{eq:new_sud_triplet}
\end{equation}
share every measured pair RDM, because the singlet and mixed triplet have the same one-site RDM $\id/2$. Only $\sigma$ has singlet support; $\rho$ is weakly invariant but not strongly symmetric. Figure~\ref{fig:new_pairs}(a) shows that no pair can be omitted universally for $\SU(2)$. The reduction for $d\ge3$ is representation-specific and assumes agreement with singlet exchange correlations.

\section{Limits for spontaneous order}
\label{sec:new_spontaneous_order}

Pair RDMs cannot determine every property of the global density matrix~\cite{Feng2025,Liu2026,Divi2026}.

One such aspect is spontaneous order, which is an asymptotic property of a sequence rather than an exact symmetry condition on a single finite state. Consider exactly strongly symmetric states $\rho_L$ on growing lattices and bounded charged operators $O_i$ of fixed support, with both the system volume and operator separation tending to infinity. The ordinary and global fidelity correlators are~\cite{Lessa2025,Wang2026}
\begin{equation}
 C_{ij}=\Tr(\rho_L O_iO_j^\dagger),\qquad
 \mathcal F_{ij}=\norm{\sqrt{\rho_L}\,O_iO_j^\dagger\sqrt{\rho_L}}_1.
 \label{eq:new_order_correlators}
\end{equation}
Here $\norm{\cdot}_1$ is the trace norm. Nondecaying $C_{ij}$ signals conventional weak-symmetry-breaking order, whereas nondecaying $\mathcal F_{ij}$ diagnoses strong-symmetry-breaking order in the global fidelity formulation. If the corresponding ordinary charged correlations decay, the latter signals strong-to-weak spontaneous symmetry breaking. These asymptotic orders can coexist with $U(g)\rho_L=\chi_L(g)\rho_L$ at every finite $L$. Related canonical-purification diagnostics use R\'enyi-1 and Wightman correlations~\cite{Weinstein2025,LiuWightman2025}.

Equal pair RDMs fix $C_{ij}$ for single-site $O_i$ and $O_j$, but need not fix composite ordinary order or the nonlinear $\mathcal F_{ij}$. This failure occurs even at sharp charge. For $N=2^m-1$ qubits, label the sites by nonzero $x\in\mathbb F_2^m$ and define binary linear patterns $c_a(x)=a\cdot x$ for nonzero $a$~\cite{Cadambe2015}. Compare their uniform mixture $\rho_{\mathrm{pat}}$ with the uniform mixture $\rho_{\mathrm{can}}$ of all configurations with charge $q=(N+1)/2$. Counting gives both ensembles the same pair RDMs and charge $q$, while distinct patterns differ on $q$ sites. Thus, for $O_i=\ket1\bra0_i$,
\begin{equation}
 \mathcal F_{ij}(\rho_{\mathrm{pat}})=0,\qquad
 \mathcal F_{ij}(\rho_{\mathrm{can}})=\frac{N+1}{4N}\longrightarrow\frac14.
 \label{eq:new_fidelity_separation}
\end{equation}
Thus exact strong symmetry and every ordinary correlation of single-site operators can agree while global fidelity order differs.

\section{Summary and discussion}
\label{sec:new_summary}

We have shown that complete two-point correlations determine exact connected onsite strong symmetry and its character without reconstructing the many-body state. Theoretically, this reduces a global density-matrix constraint to vanishing fluctuations of additive charges. Experimentally, with exact data, it extracts global-symmetry information from local observables and two-point correlations accessible in many-body systems, quantum simulators, and circuits. For Lie groups, the charge covariance kernel identifies all strong-symmetry generators. The fixed-number $\Uone$ and fundamental $\SU(d)$ examples further show that representation-specific subsets can suffice.

Future work could identify reduced sets of correlations for other onsite representations and determine what additional information is needed to diagnose strong-to-weak symmetry breaking in local many-body systems. An immediate experimental direction is to apply this construction to infer strong-symmetry groups and determine how the exact equalities used here should be replaced by finite-precision bounds.

\section{Acknowledgments}
H.Y. acknowledges support from the Japan Society for the Promotion of Science through the Grant-in-Aid for Early-Career Scientists (Grant No. JP26K17090).

\bibliography{references}
\end{document}